\documentclass[3p,preprint]{elsarticle}

\journal{Information and Computation}

\usepackage{amssymb}

\usepackage[figuresright]{rotating}

\newtheorem{theorem}{Theorem}
\newtheorem{lemma}[theorem]{Lemma}
\newtheorem{proposition}[theorem]{Proposition}
\newtheorem{corollary}[theorem]{Corollary}

\newtheorem{claim}{Claim}
\newtheorem{example}{Example}

\newproof{proof}{Proof}

\usepackage[utf8]{inputenc}
\usepackage{complexity}
\usepackage[algosection]{algorithm2e}
\usepackage{endnotes}
\usepackage{listings}
\usepackage{mathtools}
\usepackage{verbatim}
\usepackage{mathrsfs}
\usepackage{longtable}
\usepackage{enumitem}
\usepackage{etoolbox}
\usepackage{float}
\usepackage{color}
\usepackage{textcomp}
\usepackage{hyperref}

\usepackage{mathtools}

\newcommand*{\DEBUG}{}%

\ifdefined\DEBUG
\newcommand{\fixme}[1]{{\textcolor{red}{\bf{\textsf{FIXME: #1}}}}}
\newcommand{\bug}[1]{{\textcolor{blue}{\bf{\textsf{BUG: #1}}}}}
\newcommand{\idea}[1]{{\textcolor{blue}{\bf{\textsf{IDEA: #1}}}}}

\newcommand{\TODO}[1]{{\textcolor{red}{\bf{\textsf{
TODO: #1
}}}}}
\else
\newcommand{\fixme}[1]{}
\newcommand{\bug}[1]{}
\newcommand{\TODO}[1]{}
\newcommand{\idea}[1]{}
\fi

\newclass{\COMSLIP}{COM\mbox{-}SLIP}
\newclass{\COMSLIPCUP}{COM\mbox{-}SLIP^{\cup}}

\newclass{\DCM}{DCM}
\newclass{\DPDA}{DPDA}
\newclass{\PDA}{PDA}
\newclass{\DCMNE}{DCM_{NE}}
\newclass{\TwoDCM}{2DCM}
\newclass{\NCM}{NCM}
\newclass{\DPCM}{DPCM}
\newclass{\NPCM}{NPCM}
\newclass{\NPDA}{NPDA}
\newclass{\TRE}{TRE}
\newclass{\NFA}{NFA}
\newclass{\DFA}{DFA}
\newclass{\NCA}{NCA}
\newclass{\DCA}{DCA}
\newclass{\DTM}{DTM}

\newcommand\abs[1]{\left|#1\right|}

\newcommand\union{\cup}

\newcommand\natnum{\mathbb{N}}

\DeclareMathOperator{\comm}{comm}

\newcommand{\shuffle}{\hspace{1mm}{\mathbin{\mathchoice
{\rule{.3pt}{1ex}\rule{.3em}{.3pt}\rule{.3pt}{1ex}
\rule{.3em}{.3pt}\rule{.3pt}{1ex}}
{\rule{.3pt}{1ex}\rule{.3em}{.3pt}\rule{.3pt}{1ex}
\rule{.3em}{.3pt}\rule{.3pt}{1ex}}
{\rule{.2pt}{.7ex}\rule{.2em}{.2pt}\rule{.2pt}{.7ex}
\rule{.2em}{.2pt}\rule{.2pt}{.7ex}}
{\rule{.3pt}{1ex}\rule{.3em}{.3pt}\rule{.3pt}{1ex}
\rule{.3em}{.3pt}\rule{.3pt}{1ex}}\mkern2mu}}\hspace{1mm}}

\begin{document}

\begin{frontmatter}




\title{On the Complexity and Decidability of Some Problems Involving Shuffle\tnoteref{t1}}

\tnotetext[t1]{\textcopyright 2016. This manuscript version is made available under the CC-BY-NC-ND 4.0 license \url{http://creativecommons.org/licenses/by-nc-nd/4.0/}}


\author[label1]{Joey Eremondi\fnref{fn3}}
\address[label1]{Department of Information and Computing Sciences\\ Utrecht University, P.O.\ Box 80.089 3508 TB Utrecht, The Netherlands}
\ead[label1]{j.s.eremondi@students.uu.nl}

\author[label2]{Oscar H. Ibarra\fnref{fn2}}
\address[label2]{Department of Computer Science\\ University of California, Santa Barbara, CA 93106, USA}
\ead[label2]{ibarra@cs.ucsb.edu}
\fntext[fn2]{Supported, in part, by
NSF Grant CCF-1117708 (Oscar H. Ibarra).}

\author[label3]{Ian McQuillan\fnref{fn3}\corref{corr}}
\address[label3]{Department of Computer Science, University of Saskatchewan\\
Saskatoon, SK S7N 5A9, Canada}
\ead[label3]{mcquillan@cs.usask.ca}
\cortext[corr]{Corresponding author}
\fntext[fn3]{Supported, in part, by Natural Sciences and Engineering Research Council of Canada Grant 327486-2010 (Ian McQuillan).}

\begin{abstract}
The complexity and decidability of various decision problems involving the shuffle operation (denoted by $\shuffle$) are studied. The following three problems are all shown to be $\NP$-complete: given a nondeterministic finite automaton ($\NFA$) $M$, and two words $u$ and $v$, is $L(M) \not\subseteq u\shuffle v$,  is $u\shuffle v \not\subseteq L(M)$, and is $L(M) \neq u \shuffle v$? 
It is also shown that there is a polynomial-time algorithm to determine,
for $\NFA$s $M_1, M_2$, and a deterministic pushdown automaton $M_3$, whether
$L(M_1) \shuffle L(M_2) \subseteq L(M_3)$. The same is true when $M_1, M_2,M_3$
are one-way nondeterministic $l$-reversal-bounded $k$-counter machines, with $M_3$ being deterministic.
Other decidability and complexity results are presented for testing
whether given languages $L_1, L_2$, and $R$ from various languages families satisfy
$L_1 \shuffle L_2 \subseteq R$, and $R \subseteq L_1 \shuffle L_2$. Several closure results on shuffle are also shown.
\end{abstract}


\end{frontmatter}

\section{Introduction}

The shuffle operator models the natural interleaving between strings. It was introduced by Ginsburg and Spanier \cite{GinsburgSpanier}, where it was shown that context-free languages are closed under shuffle with regular languages, but not context-free languages. It has since been applied in a number of areas such as concurrency \cite{Ogden}, coding theory \cite{trajSurvey}, verification \cite{shuffleConcurrent}, database schema \cite{shuffleDatabases}, and biocomputing \cite{trajSurvey,bondFree}, and has also received considerable study in the area of formal languages. However, there remains a number of open questions, such as the long-standing problem as to whether it is decidable, given a regular language $R$ to tell
if $R$ has a non-trivial decomposition; that is, $R = L_1 \shuffle L_2$,
for some $L_1, L_2$ that are not the language consisting of only the empty word \cite{CSV01}.

This paper addresses several complexity-theoretic and decidability questions
involving shuffle. In the past, similar questions have been studied
by Ogden, Riddle, and Round \cite{Ogden}, who showed that there exists deterministic context-free
languages $L_1,L_2$ where $L_1 \shuffle L_2$ is $\NP$-complete. More recently,
L.\ Kari studied problems involving solutions to language equations of the form
$R = L_1 \shuffle L_2$, where some of $R,L_1, L_2$ are given, and the goal
is to determine a procedure, or determine that none exists, to solve for
the variable(s) \cite{langEq}. Also, there has been similar decidability problems investigated
involving shuffle on trajectories \cite{shuffleTraj}, where the patterns
of interleaving are restricting according to another language $T\subseteq \{0,1\}^*$ (a zero indicates that a letter from the first operand will be chosen next, and a one indicates a letter from the second operand is chosen).
L.\ Kari and Sos\'ik show that it is decidable, given
$L_1, L_2, R$ as regular languages with a regular trajectory set $T$, whether 
$R = L_1 \shuffle_T L_2$ (the shuffle of $L_1$ and $L_2$ with trajectory set $T$). Furthermore, if $L_1$ is allowed to be
context-free, then the problem becomes undecidable as long as, for every 
$n \in \natnum$, there is some word of $T$ with more than $n$ $0$'s (with a symmetric result if there is a context-free language on the right). This
implies that it is undecidable whether $L_1 \shuffle L_2 = R$, where $R$
and one of $L_1,L_2$ are regular, and the other is context-free.
In \cite{BordihnHolzerKutrib}, it is demonstrated that given two linear
context-free languages, it is not semi-decidable whether their
shuffle is linear context-free, and given two deterministic context-free
languages, it is not semi-decidable whether their shuffle is 
deterministic context-free.
Complexity questions involving so-called {\it shuffle languages}, which
are augmented from regular expressions by shuffle and iterated shuffle, have also been studied \cite{shufflelanguages}.
It has also been determined that it is $\NP$-hard to determine if a given string is the shuffle of
two identical strings (independently in \cite{UnshuffleSquares} and
\cite{ShuffleSquare2}).

Recently, there have been several papers involving the shuffle of two words. 
It was shown that the shuffle of two words with at least two letters
has a unique decomposition into the shuffle of words \cite{berstelwords}.
In fact, the shuffle of two words, each with at least two letters, has a unique decomposition over arbitrary
sets of words \cite{ShuffleTCS}.
Also, a polynomial-time algorithm has been developed that,
given a 
deterministic finite automaton ($\DFA$) $M$ and two words $u,v$, can test
if $u\shuffle v \subseteq L(M)$ \cite{Biegler2012}.
In the same work, an algorithm was 
presented that
takes a $\DFA$ $M$ as input and outputs a ``candidate solution'' $u,v$; this means, if $L(M)$ has a decomposition into the shuffle of two words, $u$ and $v$ must
be those two unique words. But the algorithm cannot
guarantee that $L(M)$ has a decomposition. This algorithm runs in $O(|u|+|v|)$
time, which is often far less than the size of the input $\DFA$, as $\DFA$s
accepting the shuffle of two words can be exponentially larger than the words
\cite{shuffleJALC}.
It has also been shown \cite{McquillanBiegler}
that the following problem is $\NP$-complete: given
a $\DFA$ $M$ and two words $u,v$, is it true that $L(M) \not\subseteq u \shuffle v$?

In this paper, problems are investigated involving three given languages 
$R,L_1, L_2$, and the goal is to determine decidability and complexity of 
testing if $R \not\subseteq L_1 \shuffle L_2, L_1 \shuffle L_2 \not\subseteq R$,
and $L_1 \shuffle L_2 \neq R$, depending on the language families of $L_1, L_2$ and
$R$.
In Section \ref{WordResults}, it is demonstrated that the following three problems are $\NP$-complete: to determine,
given an $\NFA$ $M$ and two words $u,v$ whether 
$u \shuffle v \not\subseteq L(M)$ is true, $L(M) \not\subseteq u \shuffle v$
is true, and $u \shuffle v \neq L(M)$ is true.
Then, the $\DFA$ algorithm from \cite{Biegler2012} that can output a ``candidate solution'' is
extended to an algorithm on $\NFA$s that operates in polynomial time, and
outputs two words $u,v$ such that if the $\NFA$ is decomposable 
into the shuffle of words, then $u \shuffle v$ is the unique solution.
And in Section \ref{LanguageResults}, decidability and the complexity of testing if $L_1 \shuffle L_2 \subseteq R$
is investigated involving more general language families. In particular, it is shown that
it is decidable in polynomial time, given $\NFA$s
$M_1, M_2$
and a deterministic pushdown automaton $M_3$, whether
$L(M_1) \shuffle L(M_2) \subseteq L(M_3)$. The same is true 
given $M_1, M_2$ that are one-way nondeterministic $l$-reversal-bounded $k$-counter machines, and $M_3$ is a one-way deterministic $l$-reversal-bounded $k$-counter machine.
However, if $M_3$ is a nondeterministic 1-counter machine that makes only one
reversal on the counter, and $M_1$ and $M_2$ are fixed $\DFA$s accepting $a^*$
and $b^*$ respectively, then the question is undecidable.
Also, if we have fixed languages $L_1 = (a+b)^*$ and $L_2 = \{\lambda\}$, and $M_3$ is an $\NFA$, then testing  whether $L_1 \shuffle L_2 \not\subseteq L(M_3)$ is $\PSPACE$-complete. Also, testing whether 
$a^* \shuffle \{\lambda\} \not\subseteq L$ is $\NP$-complete for $L$ accepted by an $\NFA$. For finite languages $L_1, L_2$, and $L_3$ accepted by an $\NPDA$,
it is $\NP$-complete to determine if $L_1 \shuffle L_2 \not\subseteq L_3$.
Results on unary languages are also provided.
In Section \ref{sec:converse}, testing $R \subseteq L_1 \shuffle L_2$ is addressed. This
is already undecidable if $R$ and $L_1$ are deterministic pushdown automata. However, it is decidable
if $L_1, L_2$ are any commutative, semilinear languages, and $R$ is a context-free language (even if augmented
by reversal-bounded counters). Then, in Section \ref{sec:closure}, several other decision problems, and some
closure properties of shuffle are investigated.

\section{Preliminaries}

We assume an introductory background in formal language theory and automata
\cite{HU}, as well as computational complexity \cite{GarryJohnson}.
We assume knowledge of pushdown automata, finite automata, and Turing machines, and we
use notation from \cite{HU}.
Let $\Sigma = \{a_1, \ldots, a_m\}$ be a finite alphabet. Then $\Sigma^*$ ($\Sigma^+$) is the set of all (non-empty) words over $\Sigma$. A language over $\Sigma$ is any $L \subseteq \Sigma^*$.
Given a language $L\subseteq \Sigma^*$, the complement of $L$, 
$\overline{L} = \Sigma^* - L$. 
The length of a word $w\in \Sigma^*$
is $|w|$, and for $a \in \Sigma$, $|w|_a$ is the number of $a$'s in $w$.

Let $\mathbb{N}$ be the positive integers, and $\mathbb{N}_0$
be the non-negative integers. For $n \in \mathbb{N}_0$, then define $\pi(n)$ to be $0$ if
$n = 0$, and $1$ otherwise.

Next, we formally define reversal-bounded counter machines \cite{Ibarra1978}.
A {\em one-way $k$-counter machine} is a tuple $M = (k,Q, \Sigma, \lhd, \delta, q_0, F)$, where
$Q,\Sigma, \lhd, q_0, F$ are respectively, the finite set of states, input alphabet, right input end-marker (not in $\Sigma$),
the initial state, and the set of final states. The transition function $\delta$ is a relation
from $Q \times (\Sigma \cup \{\lhd\}) \times \{0,1\}^k$ into $Q \times \{{\rm S}, {\rm R}\} \times \{-1,0,+1\}^k$, such
that  if $\delta(q,a,c_1, \ldots, c_k)$ contains $(p,d,d_1, \ldots, d_k)$ and $c_i = 0$ for some $i$, then $d_i \geq 0$ (this is to
prevent negative values in any counter). The symbols ${\rm S}$ and ${\rm R}$ give the direction of the input tape head, being
either {\it stay} or {\it right} respectively. Furthermore, $M$ is deterministic if $\delta$ is a partial function.
A configuration of $M$ is a tuple $(q,w,c_1, \ldots, c_k)$ indicating that $M$ is in state $q$ with $w$ (in $\Sigma^*$ or $\Sigma^*\lhd$)
as the remaining input, and $c_1, \ldots, c_k \in \mathbb{N}_0$ are the contents of the
counters. The derivation relation $\vdash_M$ is defined by,
$(q, aw, c_1, \ldots, c_k) \vdash_M (p, w', c_1 + d_1, \ldots, c_k + d_k)$, if
$(p, d, d_1, \ldots, d_k) \in \delta(q, a, \pi(c_1), \ldots, \pi(c_k))$ where $d = {\rm S}$ implies $w' = aw$, and $d = {\rm R}$
implies $w' = w$. Then $\vdash_M^*$ is the reflexive, transitive closure of $\vdash_M$. A word $w \in \Sigma^*$ is accepted by $M$
if $(q_0, w\lhd, 0, \ldots, 0) \vdash_M^* (q, \lhd, c_1, \ldots, c_k)$, for some $q \in F, c_1, \ldots, c_k \in \mathbb{N}_0$.
The language accepted by $M$, $L(M)$, is the set of all words accepted by $M$.
Essentially, a $k$-counter machine is a $k$-pushdown machine where each pushdown has one symbol plus an end-marker.
It is well known that a two counter machine
is equivalent to a deterministic Turing machine \cite{Minsky}.

In this paper, we will often restrict the counter(s) to
be reversal-bounded in the sense that each counter can only
reverse (i.e., change mode from non-decreasing to non-increasing
and vice-versa)
at most $r$ times for some given $r$.
In particular, when $r = 1$, the counter reverses only once, i.e.,
once it decrements, it can no longer increment.
Note that a counter that makes $r$ reversals
can be simulated by $\lceil \frac{r+1}{2} \rceil$ 1-reversal-bounded
counters. Closure and decidable properties of various
machines augmented with reversal-bounded counters have been studied
in the literature (see, e.g., \cite{Ibarra1978,IbarraDCFS2014,EIMInsertion2015,EIMDeletion2015}).
We will use the notation $\NCM(k,r)$ to represent $r$-reversal-bounded,
$k$-counter machines, and $\NCM$ to represent
all reversal-bounded multicounter machines.
Machines with one unrestricted pushdown, plus reversal-bounded counters
have also been studied \cite{Ibarra1978}. Then, $\NPCM(k,r)$ are
machines with one unrestricted pushdown, and $k$ $r$-reversal-bounded
counters, and $\NPCM$ are all such machines.
We use `D' in place of `N'
for the deterministic versions, e.g.,
$\DCM$, $\DCM(k,r)$, $\DPCM$, and $\DPCM(k,r)$. We use this notation
for both the classes of machines, and the families of languages
they accept.

\begin{example}
Consider the language $L = \{a^i b^j a^i b^j \mid i,j >0\}$. This
language is neither regular nor context-free. However, this language
can be accepting by a $\DCM(2,1)$ machine $M$. Indeed, $M$, on
input $a^i b^j a^k b^l$ reads the $a$'s, and increases the first
counter by one for each $a$ read. 
Then, it reads the $b$'s, and increases the
second counter for every $b$ read. Then, it reads the second section
of $a$'s, and decreases the first counter for every $a$ read, and 
verifies that it hits the final section of $b$'s when the first
counter is empty, thereby verifying that the number of $a$'s in
the first section is equal to the number in the second section.
Then, it does the same with the second section of $b$'s and the second
counter. Thus, $M$ accepts if and only if $i = k$ and $j = l$.
\end{example}

We will also use the notation below to represent common classes of automata (and languages), where any that we have not
defined are defined as in \cite{HU}:
$\NPDA$ for nondeterministic pushdown automata; $\DPDA$ for deterministic pushdown automata;
$\NCA$  for nondeterministic one counter machines with no reversal-bound;
$\DCA$ for deterministic $\NCA$s;
$\NFA$ for nondeterministic finite automata;
$\DFA$ for deterministic finite automata;
and $\DTM$ for deterministic Turing machines. As is well-known, $\NFA$s, $\NPDA$s, 
and 
$\DTM$s, accept exactly the regular languages, context-free languages, 
and recursively
enumerable languages, respectively.

A set
$Q \subseteq \mathbb{N}_0^m$ is a {\em linear set} if
there exist vectors $\vec{v_0}, \vec{v_1}, \ldots, \vec{v_n} \in \mathbb{N}_0^m$ such that 
$Q = \{\vec{v_0} + i_1 \vec{v_1} + \cdots + i_n \vec{v_n} \mid i_1, \ldots, i_n \in \mathbb{N}_0\}$. In this definition, the vector
$\vec{v_0}$ is called the constant and $\vec{v_1}, \ldots, \vec{v_n}$
are the {\em periods}. A {\em semilinear set} is a finite union
of linear sets. For semilinear sets 
$Q_1, Q_2 \subseteq \mathbb{N}^m, Q_1 + Q_2 = \{v \mid v_1 + v_2, v_1 \in Q_1, v_2 \in Q_2\}$.

The {\em Parikh map} of $w \in \Sigma^*, \Sigma = \{a_1, \ldots, a_m\}$ is the vector $\psi(w) = 
(|w|_{a_1}, \ldots, |w|_{a_m})$, and the Parikh map of $L$ is
$\psi(L) = \{ \psi(w) \mid w \in L\}$.
For a vector $\vec{v} \in \mathbb{N}_0^m$, the inverse
$\psi^{-1}(\vec{v}) = \{w \in \Sigma^* \mid \psi(w) = \vec{v}\}$, which
is extended to subsets of $\mathbb{N}_0^m$. A language is
{\em semilinear} if its Parikh map is semilinear.
The commutative closure of a language $L \subseteq \Sigma^*$ is 
$\comm(L) = \psi^{-1}(\psi(L))$, and a language $L$ is {\em commutative} if
it is equal to its commutative closure.
The family of all commutative semilinear languages is denoted by 
$\COMSLIP$, following notation developed in \cite{CrespiReghizzi}.

Let $u,v \in \Sigma^*$. The {\em shuffle} of $u$ and $v$, denoted
$u \shuffle v$ is the set $$\{u_1v_1 u_2 v_2 \cdots u_n v_n \mid u_i,v_i 
\in \Sigma^*, 1 \leq i \leq n, u = u_1 \cdots u_n, v=v_1 \cdots v_n\}.$$
This can be extended to languages $L_1,L_2 \subseteq \Sigma^*$
by $ L_1 \shuffle L_2 = \bigcup_{u \in L_1, v\in L_2} u \shuffle v$. Given $u,v \in \Sigma^*$, 
there is an obvious $\NFA$ with $(|u|+1)(|v|+1)$ states accepting $u\shuffle v$,
where each state stores a position within both $u$ and $v$. This
has been called the {\em naive $\NFA$} for $u \shuffle v$ 
\cite{shuffleJALC}. It was also mentioned in \cite{shuffleJALC} that if $u$ and $v$ are over disjoint alphabets, then the naive $\NFA$ is a $\DFA$.

An $\NFA$ $M= (Q,\Sigma, q_0,F,\delta)$ is {\em accessible} if,
for each $q \in Q$, there exists $u \in \Sigma^*$ such that
$q \in \delta(q_0,u)$. Also, $M$ is {\em co-accessible} if,
for each $q \in Q$, there exists $u \in \Sigma^*$ such that
$\delta(q,u) \cap F \neq \emptyset$. Lastly, $M$ is {\em trim} if it
is both accessible and co-accessible, and $M$ is {\em acyclic} if
$q \notin \delta(q, u)$ for every $q \in Q, u \in \Sigma^+$.

\section{Comparing Shuffle on Words to $\NFA$s}
\label{WordResults}

The results to follow in this section depend on the
following  result from \cite{McquillanBiegler}.

\begin{proposition}
\label{dfaSubsetNP}
It is $\NP$-complete to determine, given a $\DFA$ $M$ and words $u,v$ over an alphabet of at least two letters, if $L(M) \not\subseteq u \shuffle v$.
\end{proposition}

First, we note that this $\NP$-completeness can be
extended to $\NFA$s.

\begin{corollary}
\label{subsetwords}
It is $\NP$-complete to determine, given an $\NFA$ $M$ and words $u,v$ over
an alphabet of at least two letters, if $L(M) \not\subseteq u \shuffle v$.
\end{corollary}
\begin{proof}
$\NP$-hardness follows from Proposition \ref{dfaSubsetNP}.

To show it is in $\NP$, let $M$ be an $\NFA$ with state set $Q$.
Create a nondeterministic Turing machine that
guesses a word $w$ of length at most $|uv|+|Q|$, and verifies that
$w\in L(M)$ and that $w \notin u \shuffle v$ in polynomial time \cite{McquillanBiegler}.
And indeed, $L(M) \not\subseteq u \shuffle v$ if and only if
$L(M) \cap \{w \mid |w| \leq |uv| + |Q|, w\in \Sigma^*\} \not\subseteq u \shuffle v$, since any word longer than $|uv| + |Q|$
that is in $L(M)$ implies there is another
one in $L(M)$ with length between $|uv|+1$ and $|uv|+|Q|$, which is 
therefore not in $u \shuffle v$ (all words in $u \shuffle v$ are of length $|u| + |v|$).
\qed \end{proof}

Next, the reverse inclusion of Corollary \ref{subsetwords} will be examined. In contrast to the polynomial-time testability of $u \shuffle v \subseteq L(M)$ when $M$ is a $\DFA$
(\cite{Biegler2012}, with an alternate shorter proof appearing in
Proposition \ref{alternate} of this paper), testing 
$u\shuffle v \not\subseteq L(M)$ is $\NP$-complete for $\NFA$s.
\begin{proposition}
\label{wordssubset}
It is $\NP$-complete to determine, given an $\NFA$ $M$ and $u,v$
over an alphabet of at least two letters, whether
$u \shuffle v \not\subseteq L(M)$.
\end{proposition}
\begin{proof}
First, it is in $\NP$, since all words in $u\shuffle v$ are of length 
$\abs{uv}$, and so a nondeterministic Turing machine can be built that nondeterministically guesses one and 
tests if it is not in $L(M)$ in polynomial time.

For $\NP$-hardness, let $F$ be an instance of 
the 3SAT problem (a known $\NP$-complete problem \cite{GarryJohnson}) with a set
of Boolean variables $X = \{x_1, \ldots, x_p\}$, and a set of
clauses $\{c_1, \ldots, c_q\}$, where each clause has three literals.

If $d$ is a truth assignment, then $d$ is a function from $X$ to
$\{+,-\}$ (true or false). 
For a variable $x$, then $x^+$ and $x^-$ are literals. 
In particular, the literal $x^+$ is true under $d$
if and only if the variable $x$ is true under $d$. And, the literal
$x^-$ is true under $d$ if and only if the variable $x$ is false \cite{GarryJohnson}.
Let $y = \lceil \log_2 p \rceil +1$, which is enough to hold
the binary representation of any of $1, \ldots, p$.
For an integer $i$, $1 \leq i \leq p$, let $b(i)$ be the
string $1$ followed by the $y$-bit binary representation of $i$,
followed by $1$ again.

For $1 \leq i \leq p, 1 \leq j \leq q$, let $f(i,j)$ be defined as follows, where
it is a set of either one or two strings over $\{0,1\}$:
$$f(i,j) = \begin{cases} \{01 b(i)\} & \mbox{if~} x_i^+ \in c_j;\\
						 \{10 b(i)\} & \mbox{if~} x_i^- \in c_j;\\
						\{10 b(i), 01b(i)\} & \mbox{otherwise}.
						\end{cases}$$
For $1 \leq j \leq q$, let 
$F_j = f(1,j) f(2,j) \cdots f(p,j).$
Therefore, $F_j$ is the concatenation of languages (since each
$f(i,j)$ is a set of one or two strings).

We will next give the construction. 
Let $u = 1 b(1) 1 b(2)  \cdots 1 b(p)$,
and let $v = 0^p$.

Let $T = \{e_1 b(1) e_2 b(2) \cdots e_p b(p) \mid e_i \in \{10,01\}, 1 \leq i \leq p\}$, consisting of $2^p$ strings.
Clearly $T \subseteq u \shuffle v$, and also $T$ is 
a regular language, and a $\DFA$ $M_T$ can be built accepting
this language in polynomial time, as with a $\DFA$ $\overline{M_T}$ accepting 
$\overline{L(M_T)}$.

Then, make an $\NFA$ $M'$ accepting $\bigcup_{1 \leq j \leq q}F_j$. It is clear
that this $\NFA$ is of polynomial size. Then, make another $\NFA$ $M''$ accepting
$L(M') \cup L(\overline{M_T}$).
The following claim shows that deciding $u\shuffle v \not\subseteq L(M'')$ is
equivalent to deciding if there is a solution to the 3SAT instance.
\begin{claim}
The following three conditions are equivalent:
\begin{enumerate}
\item $u \shuffle v \cap \overline{L(M'')} \neq \emptyset$,
\item $T \cap \overline{L(M'')} \neq \emptyset$,
\item $F$ has a solution.
\end{enumerate}
\end{claim}
\begin{proof}
``$1 \Rightarrow 2$''.
Let $w\in u \shuffle v \cap \overline{L(M'')}$. Then $w \notin L(M'')$,
and since $L(\overline{M_T}) \subseteq L(M'')$, necessarily $w \in T$.

``$2 \Rightarrow 1$''.
Let $w \in T \cap \overline{L(M'')}$. But, $T \subseteq u \shuffle v$; and so
$w \in u \shuffle v \cap \overline{L(M'')}$.

``$2 \Rightarrow 3$''.
Assume $w \in T \cap \overline{L(M'')}$.
Thus, $w = e_1 b(1) e_2 b(2) \cdots e_p b(p), e_i \in \{10,01\}$,
but $w \notin \bigcup_{1 \leq j \leq q} F_j$.
Let $d$ be the truth assignment obtained from $w$ where
$$d(x_i) = \begin{cases} +, & \mbox{if~} e_i = 10; \\ -, & \mbox{if~} e_i = 01;\end{cases}$$ for all $i, 1 \leq i \leq p$.
Thus, for every $j$, 
$1 \leq j \leq q$, $w \notin F_j$, but for all variables
$x_i$ not in $c_j$, $e_i b(i)$ must be an infix of words in $F_j$ since $10 b(i)$ and
$01b(i)$ are both in $f(i,j)$ when $x_i$ is not in $c_j$. So one
of the words encoding the (three) variables in $c_j$, must have $10 b(i)$ as an infix of 
words in $F_j$
where $d(x_i) = +$, or $01b(i)$ as an infix of words in $F_j$ where $d(x_i) = -$,
since otherwise $F_j$ would have as infix, for each $x_i$ that is
a variable of $c_j$, 
$01 b(i) $ if $x_i^+ \in c_j$, and $10b(i)$ if $x_i^- \in c_j$, and so $w$
would be in $F_j$, a contradiction.
Thus, $d$ makes clause $c_j$ true, as is the case with every clause.
Hence, $d$ is a satisfying truth assignment, and $F$ is satisfiable.

``$3 \Rightarrow 2$''.
Assume $F$ is satisfiable, hence $d$ is a satisfying truth assignment.
Create $$w = e_1 b(1) e_2 b(2) \cdots e_p b(p),$$
where $$e_i = \begin{cases} 10 ,& \mbox{if~} d(x_i) = +; \\ 01, & \mbox{if~} d(x_i) = -.\end{cases}$$
Then $w \in T$. Also, for each $j$, $d$ applied to some
variable, say $x_i$, must be in $c_j$, but then by the construction
of $F_j$,
$e_i b(i)$ must not be an infix of any word in $F_j$. Hence,
$w \notin \bigcup_{1 \leq j \leq q}F_j$, $w \notin L(M')$, and
$w \notin L(M'')$. Hence, $w \in T \cap \overline{L(M'')}$.
\qed \end{proof}
\qed \end{proof}

Next, we examine the complexity of testing inequality between languages accepted by $\NFA$s and words of a very simple form.
\begin{proposition}
\label{thm:nfaUniverse}
It is $\NP$-complete to test, given $a^p,b^q \in \Sigma^*, p,q \in \mathbb{N}_0$, and $M$ an $\NFA$ over 
$\Sigma = \{a,b\}$, whether
$L(M) \neq a^p \shuffle b^q$.
\end{proposition}
\begin{proof}

First, it is immediate that the problem is in $\NP$, by
Corollary \ref{subsetwords} and Proposition \ref{wordssubset}.

To show $\NP$-hardness, the problem in Proposition \ref{dfaSubsetNP} is used.

Given $M$, a $\DFA$, and words $u, v$, we can construct the naive shuffle $\NFA$ $N$ for $u \shuffle v$. The naive $\NFA$ is of polynomial size in the length of $u$ and $v$.
Let $(p,q) = (\abs{uv}_a, \abs{uv}_b)$.
Then construct the naive $\NFA$ $A$ accepting $a^p \shuffle b^q$, which is 
a polynomially sized $\DFA$ since $a^p,b^q$ are over disjoint alphabets. 
Thus, another $\DFA$ can be built accepting $\overline{L(A)}$.
We can then construct an $\NFA$ $M'$ in polynomial time which accepts 
$(a^p \shuffle b^q \cap \overline{L(M)} ) \union L(N) \cup (L(M) \cap (\overline{a^p \shuffle b^q}))$ as $M$ is already a $\DFA$.
Also, $u \shuffle v \subseteq a^p \shuffle b^q$ since the latter contains all
words with $p$ $a$'s and $q$ $b$'s.

We will show 
$L(M) \subseteq u \shuffle v$ if and only if 
$L(M') = a^p \shuffle b^q$.

Assume $L(M) \subseteq L(N) (= u \shuffle v)$. Then 
$L(M) \cap (\overline{a^p \shuffle b^q}) = \emptyset$ since
$u \shuffle v \subseteq a^p \shuffle b^q$. All other words in 
$L(M')$ are in $a^p \shuffle b^q$. Thus, $L(M') \subseteq a^p \shuffle b^q$.
Let $w \in a^p \shuffle b^q$. If $w \notin L(M)$, then $w \in L(M')$. 
If $w \in L(M)$, then $w \in L(N) \subseteq L(M')$, by the
assumption.

Assume $L(M') = a^p \shuffle b^q$. Let $w \in L(M)$. Then 
$L(M) \cap (\overline{a^p \shuffle b^q}) = \emptyset$ by the assumption.
So, $L(M) \subseteq a^p \shuffle b^q$. Assume $w \in L(M)$ but
$w \notin L(N)$. However, $w \in L(M')$ by the assumption, a contradiction, as
$w \notin a^p \shuffle b^q \cap \overline{L(M)}$, and 
$w \notin L(M) \cap \overline{a^p \shuffle b^q}$, implying $w \in L(N)$.

Hence, the problem is $\NP$-complete.
\qed \end{proof}


To obtain the following corollary, it only needs to be shown that the problem is
in $\NP$, which again follows from Corollary \ref{subsetwords} 
and Proposition \ref{wordssubset}.
\begin{corollary}
It is $\NP$-complete to determine,
given an $\NFA$ $M$ and words $u,v$ over an alphabet of size at least two, 
if $L(M) \neq u \shuffle v$.
\end{corollary}

It is known that there is a polynomial-time algorithm that, given
a $\DFA$, will output two words $u$ and $v$ such that, if $L(M)$ is decomposable 
into the shuffle of two words, then this implies $L(M) = u \shuffle v$
\cite{Biegler2012}. 
Moreover, this
algorithm runs in time $O(|u| +|v|)$, which is sublinear in the input $M$.
This main result from \cite{Biegler2012} is as follows:
\begin{proposition}
\label{DFAalg}
Let $M$ be an acyclic, trim, non-unary $\DFA$ over $\Sigma$. Then we can determine words $u,v\in \Sigma^+$ such that, $L(M)$ has a shuffle decomposition into two words implies $L(M) = u \shuffle v$ is the unique decomposition. This can be calculated in $O(|u|+|v|)$ time.
\end{proposition}
However, the downside
to this algorithm is that it can output two strings $u$ and $v$, when $L(M)$ is not
decomposable. Thus, the algorithm does not check whether $L(M)$ is 
decomposable, but if it is, it can find the decomposition in time
usually far less than the number of states of the $\DFA$.
The decomposition also must be unique over words (this is always true when there are at least two combined letters used in words of $L(M)$ which is the purpose of the non-unary condition in the proposition statement) \cite{berstelwords}.

It is now shown that this result scales to $\NFA$s, while remaining polynomial time complexity. The algorithm in \cite{Biegler2012} scans at most $O(|u|+|v|)$ transitions and states
of the $\DFA$ from initial state towards final state. From an $\NFA$, it becomes possible to apply the standard subset construction \cite{HU} on the $\NFA$ only by 
creating states and transitions for the transitions and states examined
by this algorithm (thus, the $\NFA$ is never fully determinized, and only a subset of the transitions and states of the $\DFA$ are created and traversed). Because the algorithm essentially follows one ``main''
path from initial state to final state in the $\DFA$, the amount of work required for $\NFA$s
is still polynomial.

\begin{proposition}
There is a polynomial-time complexity 
algorithm that, given an acyclic, non-unary $\NFA$ $M=(Q_N, \Sigma, q_{N0}, F_N, \delta_N)$, can find strings $u,v \in \Sigma^+$, such that, $L(M)$ has a decomposition into two words implies $L(M) = u \shuffle v$ is the unique decomposition. Moreover, this algorithm runs in time 
$O((|u|+|v|) \abs{Q_N}^2)$.
\end{proposition}
{\it Proof sketch.}
Uniqueness again follows from \cite{berstelwords}.

The algorithm outputs
words $u,v\in \Sigma^+$ such that
either $L(M) = u \shuffle v$ or $M$ is not shuffle decomposable.
It is based off the one described in \cite{Biegler2012}, which is quite detailed, and thus not reproduced here, although we will refer to it.

In order to use the algorithm in Proposition \ref{DFAalg}, first all states that are not accessible or not co-accessible are removed. For this, a breadth-first graph search algorithm can be used to detect which states can be reached
from $q_0$ in $O(|Q_N|^2)$ time. It also verifies that all final states reached are the same distance from the 
initial state, and if not, $M$ is not decomposable.
Then, collapse these final states down to one state $q_f$ and remove all outgoing
transitions, which does not change the language accepted since $M$ is acyclic. Then, check
which states can be reached from $q_f$ following transitions in reverse using
the graph search, and remove all states that cannot be reached. This results
in an $\NFA$ $M_1 = (Q_1, \Sigma,q_1, \{q_f\}, \delta_1)$ that is trim and
accepts $L(M)$.

Let $M_D = (Q_D, \Sigma, q_{D0}, F_D, \delta_D)$ be the $\DFA$ obtained from $M_1$ via the subset construction (we do not compute this, but will refer to it).
Necessarily $M_D$ is trim and acyclic, since $M_1$ is as well. Then
$q_{D0} = \{q_1\}, F_D = \{ P \mid P \in Q_D, q_f \in P\}$.

We modify the algorithm of Proposition \ref{DFAalg} as follows: In place of $\DFA$ states, we use subsets of $Q_1$ from $Q_D$ \cite{HU}.
However, states and transitions are only computed as needed in the algorithm.
Any time $\delta(P,a)$ is referenced in the algorithm, we first compute the deterministic transition as follows: $\delta_D(P,a) = \bigcup_{p \in P} \delta_1(p,a)$, and then use this transition.
Since there are at most $\abs{Q_1}$ states in a subset of $Q_1$, any 
transition of $\delta_D$ defined on a given state and a given letter can transition to at most $\abs{Q_1}$
states. Then, we can compute $\delta_D(P, a)$ in $O(\abs{Q_1}^2)$ time
(for each state $p\in P$, add $\delta_1(p,a)$ into a sorted list without
duplicates). As it is making the list, it can test if this state is final by testing if $q_f \in \delta_D(P,a)$.
Therefore, this algorithm inspects $O(|u|+|v|)$ states and transitions
of $M_D$, which takes $O(|Q_1|^2(|u|+|v|))$ time to compute
using the subset construction.
\qed 

\section{Testing Inclusion of the Shuffle of Languages in Another Language}
\label{LanguageResults}

A known result involving shuffle on words is that 
there is a polynomial-time test to determine, given words $u,v\in \Sigma^+$ and a $\DFA$ $M$, whether $u \shuffle v \subseteq L(M)$ \cite{Biegler2012}.
An alternate simpler proof of this result is demonstrated next, 
and then this proof technique will be used to extend to more general decision problems.

\begin{proposition}
There is a polynomial-time algorithm to determine, 
given $u,v \in \Sigma^+$, and a $\DFA$ $M$, whether or not
$u\shuffle v \subseteq L(M)$.
\label{alternate}
\end{proposition}
\begin{proof}
Clearly, $u \shuffle v$ is a subset of $L(M)$
if and only if $L(A)  \cap \overline{L(M)} = \emptyset$, 
where $A$ is the naive $\NFA$ accepting $u \shuffle v$.
A $\DFA$ accepting $\overline{L(M)}$ can be built in polynomial time, and
the intersection is accepted by an
$\NFA$ using the standard construction \cite{HU} whose emptiness can be checked in polynomial time \cite{HU}.
\qed
\end{proof}

This result will be generalized in two ways. First, instead 
of individual words $u$ and $v$,
languages from $\NCM(k,r)$, for some fixed $k, r$ will be used. Moreover, instead of a $\DFA$ for the right side of the inclusion, a $\DCM(k,r)$
machine will be used.

\begin{proposition} \label{prop6}
Let $k, r$ be any fixed integers.  
It is decidable, given $M_1, M_2 \in \NCM(k,r)$
and $M_3 \in \DCM(k,r)$,  whether $L(M_1) \shuffle L(M_2)  \subseteq L(M_3)$.
Moreover, the decision procedure is polynomial in $n_1 + n_2 + n_3$, 
where $n_i$ is the size of $M_i$.
\end{proposition}
\begin{proof}
First, construct from $M_1$ and $M_2$, an $\NCM$ $M_4$ that accepts 
$L(M_1) \shuffle L(M_2)$.
Clearly  $M_4$ is an $\NCM(2k, r)$, and the size of $M_4$ is
polynomial in $n_1 + n_2$.

Then, construct from $M_3$ a $\DCM(k,r)$ machine $M_5$ 
accepting the complement of $L(M_3)$, which can be done in polynomial
time \cite{Ibarra1978}.

Lastly, construct from $M_4$ and $M_5$ an $\NCM(3k,r)$
machine $M_6$ accepting $L(M_4) \cap L(M_5)$
by simulating the machines in parallel.

It is immediate that $L(M_1) \shuffle L(M_2) \subseteq L(M_3)$
if and only if $L(M_6) = \emptyset$. Further, it has been
shown that for any fixed $t,s$, it is decidable in polynomial
time, given $M$ in $\NCM(t,s)$, whether $L(M) = \emptyset$ \cite{Gurari1981220}.
\qed
\end{proof}

Actually, the above proposition can be made stronger.  For
any fixed $k, r$, the decidability of non-emptiness of $L(M)$ for
an $\NCM(k,r)$ is in $\NLOG$, the class of languages
accepted by nondeterministic Turing machines in logarithmic space \cite{Gurari1981220}. 
It is known that
$\NLOG$ is contained in the class of
languages accepted by deterministic Turing machines in polynomial
time (whether or not the containment is proper is
open).  By careful analysis of the constructions
in the proof of the above proposition, $M_6$,
could be constructed by a logarithmic space deterministic Turing machine.  Hence:
\begin{corollary}
Let $k, r$ be any fixed integers.  
The problem of deciding, given $M_1, M_2 \in \NCM(k,r)$
and $M_3 \in \DCM(k,r)$,  whether $L(M_1) \shuffle L(M_2)  \subseteq L(M_3)$,
is in $\NLOG$.
\end{corollary}

Proposition \ref{prop6} also holds if $M_1$ and $M_2$ are $\NFA$s and 
$M_3$ is a deterministic pushdown automaton.

\begin{proposition} \label{Prop8}
It is decidable, given $\NFA$s $M_1, M_2$
and $M_3 \in \DPDA$,  whether $L(M_1) \shuffle L(M_2)  \subseteq L(M_3)$.
Moreover, the decision procedure is polynomial in $n_1 + n_2 + n_3$, 
where $n_i$ is the size of $M_i$.
\end{proposition}

\begin{proof}
The proof and algorithm proceeds much like the proof of Proposition \ref{prop6}. Given two $\NFA$s $M_1, M_2$, 
another $\NFA$ $M_4$ that accepts $L(M_1) \shuffle L(M_2)$ can be constructed
in polynomial time. Then, for a given $\DPDA$ $M_3$, a
$\DPDA$ $M_5$ can be constructed accepting its complement in polynomial time
(and is of polynomial size) \cite{Geller}. Also, given an $\NFA$ $M_4$
and a $\DPDA$, a $\NPDA$ $M_6$ can be built in polynomial time accepting
$L(M_4) \cap L(M_5)$. As above, $L(M_1) \shuffle L(M_2) \subseteq L(M_3)$
if and only if $L(M) = \emptyset$, and emptiness is decidable in
polynomial time for $\NPDA$s \cite{HU}.
\qed
\end{proof}

In contrast to Proposition \ref{Prop8}, the following is shown:

\begin{proposition} \label{prop7}
It is undecidable, given one-state $\DFA$s $M_1$ accepting $a^*$ 
and $M_2$ accepting $b^*$, and an $\NCM(1,1)$ machine $M_3$ over 
$\{a,b\}$, whether $L(M_1) \shuffle L(M_2) \subseteq L(M_3)$.
\end{proposition}
\begin{proof}
Let $\Sigma = \{a,b\}$.
Then $L_1 \shuffle L_2 =  \Sigma^*$.
Let $M_3  \subseteq \Sigma^*$ be an $\NCM(1,1)$ machine.
Then $L_1 \shuffle L_2 \subseteq  L_3$ if and only if $L_3 = \Sigma^*$.
The result follows, since the universality problem for
$\NCM(1,1)$ is undecidable \cite{Baker1974}.
The idea is the following: Given a single-tape deterministic
Turing machine $Z$, we construct
$M_3$ which, when  given any input $w$, accepts if and only
if $w$ does not represent a halting sequence of configurations
of $Z$ on an initially blank tape
(by guessing a configuration $ID_i$, and extracting the symbol at a nondeterministically
chosen position $j$ within this configuration, storing $j$ in the counter, and then checking
that the symbol in position $j$ in the next configuration $ID_{i+1}$ determined by decrementing the counter 
is not compatible with the next move of the $\DTM$ from $ID_i$; see \cite{Baker1974}). 
Hence, $L(M_3)$ accepts the universe
if and only if $Z$ does not halt.
By appropriate coding, the universe can be reduced to $\{a,b\}^*$. 
\qed
\end{proof}
Note that the proof of Proposition \ref{prop7} shows: Let $G$ be a language family such that universality is undecidable. 
Then it is undecidable, given one-state $\DFA$s $M_1$ accepting $a^*$
and $M_2$ accepting $b^*$, and $L$ in $G$, whether $L(M_1) \shuffle  L(M_2) \subseteq  L$.

%

\begin{proposition}
Let $L_1 = (a+b)^*$ and $L_2 = \{\lambda\}$.
It is $\PSPACE$-complete, given an $\NFA$ $M$ with input
alphabet $\{a,b\}$, whether $L_1 \shuffle L_2 \not \subseteq 
L(M)$.
\end{proposition}
\begin{proof}
Clearly,  $(a+b)^* \shuffle  \{\lambda\}    \not \subseteq L$
if and only if  $L \not = (a+b)^*$.   The result follows, since
it is known that this question is $\PSPACE$-complete
(see, e.g., \cite{GarryJohnson}).
\qed
\end{proof}

\vskip .25cm

\noindent
{\bf Remark 2.}
In Proposition \ref{prop7},
if $M_1$ and $M_2$ are $\DFA$s accepting finite languages, and $L$ is a language in any family with a decidable membership problem, then it is decidable whether $L(M_1) \shuffle L(M_2) \subseteq L$.
This is clearly true by enumerating all strings in $L(M_1) \shuffle L(M_2)$
and testing membership in $L$.

Next, shuffle over unary alphabets is considered.

\begin{proposition}
\label{unarycase}
It is decidable, given languages $L_1, L_2, L_3$ over alphabet $\{a\}$ accepted
by $\NPCM$s, whether
 $L_1  \shuffle  L_2  \subseteq L_3$.
\end{proposition}
\begin{proof}
It is known that the Parikh map of the language accepted by any $\NPCM$
is an effectively computable semilinear set (in this case over
$\mathbb{N}$) \cite{Ibarra1978}
and, hence, the languages $L_1, L_2, L_3$ can be accepted by
$\DFA$s over a unary alphabet.  
\qed
\end{proof} 

\begin{proposition}
It is $\NP$-complete to decide, for an $\NFA$ $M$ over alphabet $\{a\}$, whether
$a^* \shuffle  \{\lambda\}  \not \subseteq L(M)$.
\end{proposition}
\begin{proof}
Clearly,  $a^* \shuffle  \{\lambda\}    \not \subseteq L$
if and only if  $L \not = a^*$.   The result follows, since
it is known that this question is $\NP$-complete
(see, e.g., \cite{GarryJohnson}).
\qed
\end{proof}

For the case when the unary languages $L_1$ and $L_2$ are finite:
\begin{proposition}
It is polynomial-time decidable,
given two finite unary languages $L_1$ and $L_2$
(where the lengths of the strings in $L_1$ and $L_2$ are represented
in binary) and a unary language $L_3$ accepted by an $\NFA$ $M$, all over the same letter,
whether $L_1 \shuffle L_2 \subseteq L_3$.
\end{proposition}
\begin{proof}
Let $r$ be the sum of the cardinalities of $L_1$ and $L_2$,
$s$ be the length of the binary representation of the
longest string in $L_1 \cup L_2$,
and $t$ be the length of binary representation of $M$.

We represent the $\NFA$ $M$ by an $n \times n$
Boolean matrix $A_M$, where $n$ is the number of
states of $M$, and $A_M(i,j) = 1$ if there is a transition
from state $i$ to state $j$; $0$ otherwise.

Let $x$ be the binary representation of a unary string $a^d$,
where $d = d_1 + d_2$, $a^{d_1} \in L_1$, and $a^{d_2} \in L_2$.
To determine if $a^d$ is in $L_3$, we compute 
$A_M^d$ and check that for some accepting state $p$, the $(1,p)$
entry is $1$.  Since the computation of $A_M^d$ can be
accomplished in $O(\log~d)$ Boolean matrix multiplications
(using the ``right-to-left binary method for exponentiation'' 
technique used to compute
$x^m$ where $m$ is a positive integer in $O(\log~m)$ multiplications, described in Section 4.6.3 of
\cite{knuth2}), and since matrix multiplication can be calculated in
polynomial time,
it follows that we can decide whether $L_1 \shuffle L_2 \subseteq L_3$
in time polynomial in $r + s +t$.
\qed
\end{proof}

However, when the alphabet of the finite languages 
$L_1, L_2$ is at least binary:
\begin{proposition}
It is $\NP$-complete to determine, given finite languages $L_1$ and $L_2$,
and an $\NPDA$ $M$ accepting $L_3$, whether 
$L_1 \shuffle L_2 \not\subseteq L_3$.
\end{proposition}
\begin{proof}
$\NP$-hardness follows from Proposition \ref{wordssubset}. To show that it is in $\NP$,
guess a word $u\in L_1$, and $v \in L_2$, guess a word $w$ of
length $|u|+|v|$, and verify that it is in $u \shuffle v$ \cite{McquillanBiegler}. Then,
verify that $w\notin L_3$, which  can be done in polynomial time since the
membership problem for $\NPDA$s can be solved
in polynomial time.
\qed
\end{proof}

\begin{proposition}
\label{DPDAonedirection}
It is undecidable, given languages $L$ and $L_1$ accepted by 1-reversal-bounded $\DPDA$s
(resp., $\DCA$s), whether $L_1 \shuffle \{\lambda\} \subseteq L$.
\end{proposition}
\begin{proof}
It is known that the disjointness problem for
1-reversal-bounded $\DPDA$s (resp., $\DCA$s) is undecidable 
(see, e.g., \cite{Ibarra1978}).
Let $L, L_1 \in \DPDA$. Then $L \cap L_1 = \emptyset$ if
and only if $L_1 \subseteq \overline{L}$ if and only if
$L_1 \shuffle \{\lambda\} \subseteq \overline{L}$, and 
$\overline{L}$ is in $\DPDA$ as it is closed under complement.
\qed
\end{proof}

\section{Testing Inclusion of a Language in the Shuffle of Languages}
\label{sec:converse}

In this section, the reverse containment is addressed. That is,
given $L, L_1, L_2$, is $L \subseteq L_1 \shuffle L_2$? This
question will depend on the language families where
$L, L_1$, and $L_2$ belong.

First, for 1-reversal-bounded $\DPDA$s, the following is immediate.
The proof is identical to that of Proposition
\ref{DPDAonedirection}.
\begin{proposition}
It is undecidable, given languages $L$ and $L_1$
accepted by 1-reversal-bounded $\DPDA$s (resp., $\DCA$s), whether
$L \subseteq L_1 \shuffle \{\lambda\}$.
\end{proposition}

The following proposition follows from the known undecidability of the universality problem for $\NCM(1,1)$ \cite{Baker1974}.
\begin{proposition} If $L_1 \in \NCM(1,1), L_2 = \{\lambda\}$, and $L$ is the fixed regular language $\Sigma^*$, then it is undecidable whether $L \subseteq L_1 \shuffle L_2$.
\end{proposition}

\noindent
{\bf Remark 3:} 
We are currently examining the question of whether
it is undecidable, given a regular language $L$ and languages
$L_1, L_2$ accepted by $\DPDA$s (resp., $\DCA$s, $\DPCM$s), whether
$L \subseteq L_1 \shuffle L_2$.  In particular, we are 
interested in the simple case when $L$ and $L_1$ are
regular and $L_2$ is in $\DCM(1,1)$.

\vskip .25cm

Next, we examine some families where decidability occurs.
The following is true since the shuffle of regular languages is regular, and the containment problem is decidable for regular languages.

\begin{proposition}
If $L_1, L_2, L$ are all regular languages, then it is decidable whether $L \subseteq L_1 \shuffle L_2$.
\end{proposition}

The following is true since the shuffle of an $\NPCM$ and an $\NCM$ is an $\NPCM$ language, and containment can be decided by using the decidable membership problem for $\NPCM$ \cite{Ibarra1978}.

\begin{proposition} If $L_1 \in \NPCM, L_2 \in \NCM$, and $L$ is finite, then it is decidable whether $L \subseteq L_1 \shuffle L_2$.
\end{proposition}

\begin{proposition} If $L \in \NPCM$, and $L_1, L_2 \in \DCM$ over disjoint alphabets, then it is decidable whether $L \subseteq L_1 \shuffle  L_2$.
\end{proposition}
\begin{proof} 
Let $M_1, M_2 \in \DCM$ over disjoint alphabets, such that
$L_1 = L(M_1)$ and $L_2 = L(M_2)$, where $M_1$ has $k_1$ counters and
$M_2$ has $k_2$ counters. Then $L_1 \shuffle L_2$ can be accepted by
a $k_1+k_2$ counter machine (by simulating $M_1$ on the first $k_1$
counters and $M_2$ on the last $k_2$ counters). Furthermore,
since $L_1, L_2$ are over disjoint alphabets, $L_1 \shuffle L_2$ is in $\DCM$ as well. Then we can construct $\overline{L_1 \shuffle L_2}$ since $\DCM$ is closed under complement \cite{Ibarra1978}, and test if $L \cap \overline{L_1 \shuffle L_2} = \emptyset$, which is true if and only if $L \subseteq L_1 \shuffle L_2$.
\qed
\end{proof}


The following can also be shown with a proof identical to Proposition \ref{unarycase}.
\begin{proposition}
It is decidable, given languages $L_1, L_2, L$ over alphabet $\{a\}$ accepted
by $\NPCM$s, whether:
$L  \subseteq L_1 \shuffle L_2$.
\end{proposition}

Lastly, a large family is presented for which decidability
follows. These questions will be examined next for 
commutative languages. First, two lemmas are needed.
\begin{lemma} \label{lemma1}
Let $\Sigma = \{a_1, \ldots, a_m\}$.
We can effectively construct, given semilinear
sets $Q_1, Q_2 \subseteq \mathbb{N}^m$,
a semilinear set $Q$ such that $Q = Q_1 + Q_2$.
Further, we can construct a $\DCM$ $M_Q$ to accept
$\psi^{-1}(Q)$.
\end{lemma}
\begin{proof}
It is known that every $\COMSLIP$ language is in $\DCM$ \cite{IMeDCM2015}.
Therefore, for $i = 1,2$, a $\DCM$ $M_{Q_i}$ can
be constructed
accepting $\psi^{-1}(Q_i)$.

Then construct an $\NCM$ $M$ to accept 
$\{a_1^{k_1} \cdots a_m^{k_m} ~|~k_1 = r_1+ s_1, \ldots, k_m = r_m +s_m,
a_1^{r_1} \cdots a_m^{r_m} \in \psi^{-1}(Q_1), 
a_1^{s_1} \cdots a_m^{s_m} \in \psi^{-1}(Q_2) \}$ as follows:
given input 
$w = a_1^{k_1} \cdots a_m^{k_m}$, $M$ reads the input
and nondeterministically guesses the decompositions
of the $k_i$'s into $r_i$'s and $s_i$'s, and stores them in $2m$ counters
which we call $c_1, \ldots, c_m, d_1, \ldots, d_m$ (they store the numbers $r_1, \ldots, r_m,
s_1, \ldots, s_m$).  Then $M$ simulates the computation of
$M_{Q_1}$ on input $a_1^{r_1} \cdots a_m^{r_m}$
(by decrementing the counters $c_1, \ldots, c_m$ which
have values $r_1, \ldots, r_m$, corresponding
to reading an input letter of input $M_{Q_1}$) and if $M_{Q_1}$
accepts, $M$ then simulates the computation of
$M_{Q_2}$ on input $a_1^{s_1} \cdots a_m^{s_m}$
using counters $d_1, \ldots, d_m$ similarly.  Then $M$ accepts if
$M_{Q_2}$ accepts. Since $\NCM$ only accepts semilinear languages, 
it follows that there is a semilinear set $Q$ such that $\psi(L(M)) = Q$, and hence $\psi^{-1}(Q)$ 
can be accepted by a $\DCM$.
\qed
\end{proof}

The next result follows from the definition 
of commutative semilinear languages and the previous lemma.

\begin{lemma} \label{lemma2}
Let $Q_1, Q_2 \subseteq \mathbb{N}^m$
be semilinear sets, and $Q = Q_1 + Q_2$.
Then $\psi^{-1}(Q) = \psi^{-1}(Q_1) \shuffle \psi^{-1}(Q_2)$.
Moreover, $Q$ can be effectively constructed
from $Q_1$ and $Q_2$.
\end{lemma}

We can now prove:
\begin{proposition}
Let $\Sigma = \{a_1, \ldots, a_m\}$, $m \ge 1$.
It is decidable, given an $\NPCM$ $M$ accepting
a language $L \subseteq \Sigma^*$, and $\COMSLIP$ languages 
$L_1,L_2 \subseteq \Sigma^*$ (effectively semilinear),
whether $L \subseteq L_1 \shuffle L_2$. Further,
for $L, L_1, L_2 \in \NPCM$, it is decidable whether
$L \subseteq \comm(L_1) \shuffle \comm(L_2)$.
\end{proposition}
\begin{proof}
The second statement follows from the first since $\NPCM$ is effectively
semilinear and both $\comm(L_1)$ and $\comm(L_2)$ are in $\COMSLIP$.

Using the semilinear sets $Q_1 = \psi(L_1), Q_2 = \psi(L_2)$, then from
Lemma \ref{lemma1} and Lemma \ref{lemma2}, we
can construct a $\DCM$ $M_1$ accepting $L_1 \shuffle L_2$.
We can then construct a $\DCM$ $M_2$ accepting $\overline{L(M_1)}$,
since $\DCM$ is closed under complementation.
Then we construct an $\NPCM$ $M_3$ (simulating $M$
and $M_2$ in parallel) accepting $L(M) \cap L(M_2)$.
The result follows, since
$L(M) \subseteq L_1 \shuffle L_2$
if and only if $L(M_3)= \emptyset$, which is decidable,
since the emptiness problem for $\NPCM$s is decidable \cite{Ibarra1978}.
\qed
\end{proof}
This result generalizes from $\NPCM$ to other effectively semilinear language families, such as the nondeterministic versions of the semilinear automata models from \cite{HIKS}. 

The reverse inclusion of the above proposition is not true however, since if $L_1 = (a+b)^*$ and  $L_2 = \{ \lambda \}$ (trivially commutative semilinear languages), then it is undecidable, given an $\NCM(1,1)$ $M$
with input alphabet $\{a,b\}$,
whether $L_1 \shuffle L_2  \subseteq  L(M)$. This is because it is undecidable whether an $\NCM(1,1)$ machine is equal to $\{a,b\}^*$.

\section{Other Decision and Closure Properties Involving Shuffle}
\label{sec:closure}

The following proposition follows from
the proof of Theorem 6 in \cite{BordihnHolzerKutrib}.
\begin{proposition}
It is undecidable, given two languages
accepted by 1-reversal-bounded $\DPDA$s (resp., $\DCA$s),
whether their shuffle is
accepted by a 1-reversal-bounded $\DPDA$ (resp., $\DCA$).
\end{proposition}

We can show the undecidability of a related problem:

\begin{proposition}
It is undecidable, given two 1-reversal-bounded $\DPDA$s $M_1, M_2$
(resp., $\DCA$s) and a 2-state $\DFA$ $M$,
whether $L(M) \cap  (L(M_1) \shuffle L(M_2)) = \emptyset$.
\end{proposition}
\begin{proof}
Let $L_1$ and $L_2$ be accepted by 1-reversal-bounded $\DPDA$s
(resp., $\DCA$s) over input alphabet $\Sigma$.
Let $\Sigma' = \{a' ~|~ a \in \Sigma \}$.
Define the homomorphism $h$ by: $h(a) = a'$ for all $a \in \Sigma$.
Let $L_3 = \{ w'  ~|~ w' = h(w), w \in L_1 \}$.
(Thus, $L_3$ is a primed version of $L_1$.)
Let $L = \{a_1a'_1 \cdots a_k a'_k  ~|~  k \ge 0,
a_1, \ldots, a_k \in \Sigma \}$.
Clearly, $L$ is regular and can be accepted by a 2-state $\DFA$.
Now  $L \cap (L_1 \shuffle L_3) = \emptyset$
if and only if $L_1 \cap L_2 = \emptyset$, which is undecidable.
\qed
\end{proof}

Now, we consider the shuffle of languages from various families, and contrast closure results
with established results on commutative languages.

Clearly, if $L_1, L_2$ are in $\NCM$, then
$L_1 \shuffle L_2$ is also in $\NCM$.  Thus, $\NCM$ is closed
under shuffle.  However, we have:

\begin{proposition}\label{propnew1}
There are languages $L_1, L_2 \in \DCM(1,1)$
such that $L_1 \shuffle L_2$ is over a two letter alphabet, but not 
a context-free language ($\NPDA$).
\end{proposition}
\begin{proof}
Let $L_1 = \{a^n ba b^n a \mid n >0 \}$ and $L_2 = \{b^m a^{m+1} \mid m> 0 \}$.
If $L = L_1 \shuffle L_2$ is a context-free language, then
$L' = L \cap \{a^i b^j a^k b^l a^p \mid i,l \geq 1, j,p>1, k=2\} 
= \{a^n b^{m+1} a^2 b^n a^{m+1} \mid n, m  >0 \}$ is
also a context-free language. But it is easy to show, using the Pumping Lemma,
that $L'$ is not a context-free language.
\qed
\end{proof}
This contrasts the commutative case as $L_1\shuffle L_2$ is over a two letter alphabet and is semilinear, but
not a context-free language, but every two letter semilinear language that is commutative
is a context-free language.

It is known that the commutative closure of every $\NPCM$ (or effectively semilinear language family) is a $\DCM$ language \cite{IMeDCM2015}. By contrast, the shuffle
of two 1-reversal $\DPDA$s can be significantly more complex, and might not even be an $\NPCM$ language, and can
create non-semilinear languages.
\begin{proposition} \label{propnew2}
There are languages $L_1, L_2$ accepted by 1-reversal-bounded $\DPDA$s
(resp., $\DCA$s)
such that $L_1 \shuffle L_2$ is not in $\NPCM$.
\end{proposition}
\begin{proof}
For Part 1, let
\begin{eqnarray*}
L_1 & = & \{a^{i_1} \# a^{i_3} \# a^{i_5} \# \cdots \# a^{i_{2n-1}} \$
a^{i_{2n}} \#  \cdots \# a^{i_6} \# a^{i_4} \# a^{i_2} \mid
\begin{array}[t]{l} 
n \ge 3, i_1 = 1,  i_{j+1} =  \\ i_j +1  \mbox{~for odd~} j \}, \end{array} \\
L_2 & = & \{a^{i_1} \# a^{i_3} \# a^{i_5} \# \cdots \# a^{i_{2n-1}} \$
a^{i_{2n}} \#  \cdots \# a^{i_6} \# a^{i_4} \# a^{i_2} \mid
\begin{array}[t]{l} 
n \ge 3,  i_1 = 1, i_{j+1} = \\  i_j +1 \mbox{~for even~} j \}. \end{array}
\end{eqnarray*}
Clearly, $L_1, L_2$ can be accepted by 1-reversal $\DPDA$s.
Let $L = L_1 \cap L_2$.  Then the Parikh
map of $L$ is not semilinear
since it has the same Parikh map as the language
$\{ a^1 \# a^2 \# a^3\# \cdots \$ a^{2n} \mid n \geq 3  \}$.
Hence, $L$ cannot be accepted by an $\NPCM$,
since it is known that the Parikh map of any $\NPCM$ language
is semilinear \cite{Ibarra1978}.

Now let $L_1, L_2$ be over alphabet $\Sigma = \{a, \#, \$ \}$.
Let $\Sigma' = \{a', \#', \$' \}$  be the ``primed''
copy of  $\Sigma$.  For any $x \in \Sigma^*$,
let $x'$ be primed version of $x$ (i.e., 
the symbols in $x$ are replaced by their primed copies).
Let $L_2' = \{x' ~|~ x \in L_2\}$.
Clearly, we can construct a 1-reversal-bounded $\DPDA$ accepting $L_2'$.

Suppose $L_3 = L_1 \shuffle L_2'$ can be accepted by
an $\NPCM$.   We can also then construct an $\NPCM$
accepting $L_4 = L_3 \cap  \{ss' ~|~ s \in \Sigma \}^*$.
Let $h$ be a homomorphism that maps the primed symbols
to $\lambda$ (i.e., they are erased) and leaves the un-primed
symbols unchanged.
Since $\NPCM$ languages are closed under homomorphism,
$h(L_4)$ can also be accepted by an $\NPCM$.
This gives a contradiction, since
$h(L_4) = L = L_1 \cap L_2$ cannot be accepted by an $\NPCM$. 

The proof for Part 2 is similar, except that we modify
the languages  $L_1, L_2$, as follows:
\begin{eqnarray*}
L_1 & = & \{a^{i_1} \# a^{i_2} \# a^{i_3} \# a^{i_4} \# a^{i_5} \# a^{i_6} \#
\cdots \# a^{i_{2n-1}} \# a^{i_{2n}} \mid \begin{array}[t]{l} 
n \ge 3, i_1 = 1, \\ i_{j+1} = i_j +1 \mbox{~for odd~} j \}, \end{array}\\
L_2 & = & \{a^{i_1} \# a^{i_2} \# a^{i_3} \# a^{i_4} \# a^{i_5} \# a^{i_6} \#
\cdots \# a^{i_{2n-1}} \# a^{i_{2n}} \mid \begin{array}[t]{l}
n \ge 3, i_1 = 1, \\ i_{j+1} = i_j +1 \mbox{~for even~} j \}.\end{array}
\end{eqnarray*}
Clearly, $L_1,L_2$ can be accepted by $\DCA$s.  The rest of the proof
is similar to that of Part 1.
\qed
\end{proof}

If in the statement of Proposition \ref{propnew2}, one of $L_1, L_2$
is accepted by an $\NCM$, then the proposition is no longer true, since it
can be easily verified that $L_1 \shuffle L_2$ could then be accepted 
by an $\NPCM$.  In contrast, for $\NCM$, 
there are languages accepted by deterministic counter automata $L_1$,
such that $L_1 \shuffle \{\lambda\}$ is not in $\NCM$.

We need the following lemma.

\begin{lemma} \label{lemnew1}
$L = \{a^n b^n ~|~ n > 0 \}$ is in $\DCM(1,1)$,
but $L^+$  (and, hence, also $L^*$) is not in $\NCM$.
\end{lemma}
\begin{proof}
It is obvious that $L$ is in $\DCM(1,1)$.  Suppose $L^+$ can be
accepted by an $\NCM$  $M$.  Consider the following languages:

$L_1 = \{a^n b^{n+1} ~|~ n > 0 \}^+ \{a^m ~|~ m > 0 \}$,

$L_2 = \{a^1\}\{b^n a^{n+1} ~|~ n >0\}^+ $.

\noindent
Clearly, we can construct from $M$, $\NCM$s $M_1$ and $M_2$
accepting $L_1$ and $L_2$, respectively.
Since the family of  $\NCM$ languages is closed under intersection,
$L _3 = L_1 \cap L_2 = \{ a^1 b^2a^3 b^4  \cdots  a^n b^{n+1} a^{n+2} ~|~ n > 0 \}$
is also in $\NCM$.  The result follows, since
the Parikh map
of any $\NCM$ language is semilinear \cite{Ibarra1978}, but
the Parikh map of $L_3$ is not semilinear.
\qed
\end{proof}

\begin{proposition}
There is a language $L_1$ ($L^+$ in Lemma \ref{lemnew1}) accepted by a $\DCA$ such that
$L_1 \shuffle \{\lambda\}$ is not in $\NCM$.
\end{proposition}
Therefore, commutative closure creates much simpler languages
than just the identity operation.
%

\section{Conclusions and Open Problems}

We investigated the complexity and decidability of various decision problems involving the shuffle operation.  In particular,
we showed that the following three problems are $\NP$-complete
for a given $\NFA$ $M$, and two words $u$ and $v$:

-- Is  $L(M) \not\subseteq u\shuffle v$?

-- Is $u\shuffle v \not\subseteq L(M)$?

-- Is $L(M) \neq u \shuffle v$? 

\noindent
We showed that there is a polynomial-time algorithm to determine,
for $\NFA$s $M_1, M_2$, and a deterministic pushdown automaton $M_3$, whether
$L(M_1) \shuffle L(M_2) \subseteq L(M_3)$. The same is true when $M_1, M_2,M_3$
are one-way nondeterministic $l$-reversal-bounded $k$-counter machines, with $M_3$ being deterministic.

We also presented decidability and complexity results for testing
whether given languages $L_1, L_2$, and $L$ from various languages
families satisfy $L_1 \shuffle L_2 \subseteq L$.

We obtained some closure properties of the shuffle operation
on languages.  In particular, we proved:

\begin{enumerate}
\item
There are languages $L_1, L_2 \in \DCM(1,1)$
such that $L_1 \shuffle L_2$ is over a two letter alphabet, but is not 
a context-free language ($\NPDA$).
\item
There are languages $L_1, L_2$ accepted by 1-reversal-bounded $\DPDA$s
(resp., $\DCA$s)
such that $L_1 \shuffle L_2$ is not in $\NPCM$.
\item
There is a language $L_1$ accepted by a $\DCA$ such that
$L_1 \shuffle \{\lambda\}$ is not in $\NCM$.
\end{enumerate}

\noindent
There are a number of remaining open problems including
the following:

\begin{enumerate}
\item
Is it undecidable, given a regular language $L$ and languages
$L_1, L_2$ accepted by $\DPDA$s (resp., $\DCA$s, $\DPCM$s), whether
$L \subseteq L_1 \shuffle L_2$?  In particular, we are 
interested in the simple case when $L$ and $L_1$ are
regular and $L_2$ is in $\DCM(1,1)$. 
\item
Same as item (1) above except that now the question is
whether $L = L_1 \shuffle L_2$.
\item What is the complexity of testing, given a $\DFA$ $M$, and
words $u,v$, whether $L(M) \neq u \shuffle v$?
\end{enumerate}

\bibliography{shuffle_refs}{}

\begin{thebibliography}{10}
\expandafter\ifx\csname url\endcsname\relax
  \def\url#1{\texttt{#1}}\fi
\expandafter\ifx\csname urlprefix\endcsname\relax\def\urlprefix{URL }\fi
\expandafter\ifx\csname href\endcsname\relax
  \def\href#1#2{#2} \def\path#1{#1}\fi

\bibitem{GinsburgSpanier}
S.~Ginsburg, E.~H. Spanier, Mappings of languages by two-tape devices, J. ACM
  12~(3) (1965) 423--434.

\bibitem{Ogden}
W.~F. Ogden, W.~E. Riddle, W.~C. Round, Complexity of expressions allowing
  concurrency, in: Proceedings of the 5th ACM SIGACT-SIGPLAN Symposium on
  Principles of Programming Languages, POPL '78, ACM, New York, NY, USA, 1978,
  pp. 185--194.

\bibitem{trajSurvey}
M.~Domaratzki, More words on trajectories, Bulletin of the {EATCS} 86 (2005)
  107--145.

\bibitem{shuffleConcurrent}
G.~S. Avrunin, U.~A. Buy, J.~C. Corbett, Integer programming in the analysis of
  concurrent systems, in: Proceedings of the 3rd International Workshop on
  Computer Aided Verification, CAV '91, Springer-Verlag, London, UK, UK, 1992,
  pp. 92--102.

\bibitem{shuffleDatabases}
W.~Gelade, W.~Martens, F.~Neven, Optimizing schema languages for {XML}:
  Numerical constraints and interleaving, SIAM Journal on Computing 38~(5)
  (2009) 2021--2043.
\newblock \href {http://dx.doi.org/10.1137/070697367}
  {\path{doi:10.1137/070697367}}.

\bibitem{bondFree}
L.~Kari, S.~Konstandinidis, P.~Sos\'ik, On properties of bond-free {DNA}
  languages, Theoretical Computer Science 334 (2005) 131--159.

\bibitem{CSV01}
C.~C{\^a}mpeanu, K.~Salomaa, S.~V{\'a}gv\"olgyi, Shuffle quotient and
  decompositions, in: W.~Kuich, G.~Rozenberg, A.~Salomaa (Eds.), Lecture Notes
  in Computer Science, Vol. 2295 of 5th International Conference on
  Developments in Language Theory, DLT 2001, Wien, Austria, Springer, Wien,
  Austria, 2001, pp. 186--196.

\bibitem{langEq}
L.~Kari, On language equations with invertible operations, Theoretical Computer
  Science 132~(1--2) (1994) 129--150.

\bibitem{shuffleTraj}
L.~Kari, P.~Sos\'ik, Aspects of shuffle and deletion on trajectories,
  Theoretical Computer Science 332~(1--3) (2005) 47--61.

\bibitem{BordihnHolzerKutrib}
H.~Bordihn, M.~Holzer, M.~Kutrib, Some non-semi-decidability problems for
  linear and deterministic context-free languages, in: M.~Domaratzki,
  A.~Okhotin, K.~Salomaa, S.~Yu (Eds.), Proceedings of the 10th International
  Conference on Implementation and Application of Automata (CIAA), Vol. 3317 of
  Lecture Notes in Computer Science, Springer Berlin Heidelberg, 2005, pp.
  68--79.

\bibitem{shufflelanguages}
J.~J\c{e}drzejowicz, A.~Szepietowski, Shuffle languages are in {{\bf P}},
  Theoretical Computer Science 250 (2001) 31--53.

\bibitem{UnshuffleSquares}
S.~Buss, M.~Soltys, Unshuffling a square is {NP}-hard, Journal of Computer and
  System Sciences 80~(4) (2014) 766--776.

\bibitem{ShuffleSquare2}
R.~Rizzi, S.~Vialette, On recognizing words that are squares for the shuffle
  product, in: A.~A. Bulatov, A.~M. Shur (Eds.), Lecture Notes in Computer
  Science, Vol. 7913 of 8th International Computer Science Symposium in Russia,
  CSR 2013, Ekaterinburg, Russia, 2013, pp. 235--245.

\bibitem{berstelwords}
J.~Berstel, L.~Boasson, Shuffle factorization is unique, Theoretical Computer
  Science 273 (2002) 47--67.

\bibitem{ShuffleTCS}
F.~Biegler, M.~Daley, M.~Holzer, I.~McQuillan, On the uniqueness of shuffle on
  words and finite languages, Theoretical Computer Science 410 (2009)
  3711--3724.

\bibitem{Biegler2012}
F.~Biegler, M.~Daley, I.~McQuillan, {Algorithmic decomposition of shuffle on
  words}, Theoretical Computer Science 454 (2012) 38--50.

\bibitem{shuffleJALC}
M.~Daley, F.~Biegler, I.~McQuillan, On the shuffle automaton size for words,
  Journal of Automata, Languages and Combinatorics 15 (2010) 53--70.

\bibitem{McquillanBiegler}
F.~Biegler, I.~McQuillan, On comparing deterministic finite automata and the
  shuffle of words, in: M.~Holzer, M.~Kutrib (Eds.), Proceedings of the 19th
  International Conference on Implementation and Application of Automata
  (CIAA), Vol. 8587 of Lecture Notes in Computer Science, Springer
  International Publishing, 2014, pp. 98--109.

\bibitem{HU}
J.~E. Hopcroft, J.~D. Ullman, Introduction to Automata Theory, Languages, and
  Computation, Addison-Wesley, Reading, MA, 1979.

\bibitem{GarryJohnson}
M.~R. Garey, D.~S. Johnson, Computers and Intractability: A Guide to the Theory
  of NP-Completeness, Series of books in the mathematical sciences, W. H.
  Freeman and Company, New York, 1979.

\bibitem{Ibarra1978}
O.~H. Ibarra, Reversal-bounded multicounter machines and their decision
  problems, Journal of the ACM 25~(1) (1978) 116--133.

\bibitem{Minsky}
M.~L. Minsky, Recursive unsolvability of {P}ost's problem of ``tag'' and other
  topics in theory of {T}uring machines, Annals of Mathematics 74~(3) (1961)
  437--455.

\bibitem{IbarraDCFS2014}
O.~H. Ibarra, Automata with reversal-bounded counters: A survey, in:
  H.~J\"urgensen, J.~Karhum\"aki, A.~Okhotin (Eds.), Descriptional Complexity
  of Formal Systems, Vol. 8614 of Lecture Notes in Computer Science, Springer
  International Publishing, 2014, pp. 5--22.

\bibitem{EIMInsertion2015}
J.~Eremondi, O.~H. Ibarra, I.~McQuillan, Insertion operations on deterministic
  reversal-bounded counter machines, in: A.-H. Dediu, E.~Formenti,
  C.~Mart\'in-Vide, B.~Truthe (Eds.), Lecture Notes in Computer Science, Vol.
  8977 of 9th International Conference on Language and Automata Theory and
  Applications, LATA 2015, Nice, France, 2015, pp. 200--211.

\bibitem{EIMDeletion2015}
J.~Eremondi, O.~H. Ibarra, I.~McQuillan, Deletion operations on deterministic
  families of automata, in: R.~Jain, S.~Jain, F.~Stephan (Eds.), Lecture Notes
  in Computer Science, Vol. 9076 of 12th Annual Conference on Theory and
  Applications of Models of Computation, TAMC 2015, Singapore, 2015, pp.
  388--399.

\bibitem{CrespiReghizzi}
S.~Crespi{-}Reghizzi, P.~S. Pietro, Commutative languages and their composition
  by consensual methods, in: Proceedings 14th International Conference on
  Automata and Formal Languages, {AFL} 2014, Szeged, Hungary, 2014, pp.
  216--230.

\bibitem{Gurari1981220}
E.~M. Gurari, O.~H. Ibarra, The complexity of decision problems for finite-turn
  multicounter machines, Journal of Computer and System Sciences 22~(2) (1981)
  220--229.

\bibitem{Geller}
M.~M. Geller, H.~B. {Hunt III}, T.~G. Szymanski, J.~D. Ullman, Economy of
  description by parsers, dpda's, and pda's, Theoretical Computer Science 4
  (1977) 143--153.

\bibitem{Baker1974}
B.~S. Baker, R.~V. Book, Reversal-bounded multipushdown machines, Journal of
  Computer and System Sciences 8~(3) (1974) 315--332.

\bibitem{knuth2}
D.~E. Knuth, Seminumerical Algorithms, 3rd Edition, Vol.~2 of The Art of
  Computer Programming, Addison-Wesley, Reading, Massachusetts, 1998.

\bibitem{IMeDCM2015}
O.~Ibarra, I.~McQuillan, The effect of end-markers on counter machines and
  commutativity, Theoretical Computer Science TBD.
\newblock \href {http://dx.doi.org/http://dx.doi.org/10.1016/j.tcs.2016.02.034}
  {\path{doi:http://dx.doi.org/10.1016/j.tcs.2016.02.034}}.

\bibitem{HIKS}
T.~Harju, O.~Ibarra, J.~Karhum\"{a}ki, A.~Salomaa, Some decision problems
  concerning semilinearity and commutation, Journal of Computer and System
  Sciences 65 (2002) 278--294.

\end{thebibliography}
\bibliographystyle{elsarticle-num}

\end{document}